\documentclass[aps,pra,superscriptaddress,showkeys]{revtex4-2}
\usepackage[paperwidth=210mm,paperheight=297mm,centering,hmargin=2cm,vmargin=2cm]{geometry}
\usepackage[american]{babel}
\usepackage{graphicx}
\usepackage{amsmath,amsthm,amsfonts,amssymb}
\usepackage{color}
\usepackage[colorlinks]{hyperref }
\newtheorem{rem}{Remark}[section]
\newtheorem{theorem}{Theorem}[section]
\newtheorem{prop}{Proposition}[section]

\newtheorem{corollary}{Corollary}[section]
\newcommand{\nZ}{{n \in \mathbb{Z}}}
\newcommand{\im}{\mathrm{Im}}

\begin{document}
\title{Persistence of integrable wave dynamics in the Discrete Gross--Pitaevskii equation: the focusing case}
\author{G. Fotopoulos}
\email{georgios.fotopoulos@actvet.gov.ae}
\affiliation{Academic Support Department, Abu Dhabi Polytechnic, P.O. Box 111499, Abu Dhabi, UAE}
\author{N.\,I. Karachalios}
\email{karan@uth.gr}
\affiliation{Department of Mathematics, University of Thessaly, Lamia 35100, Greece}
\author{V. Koukouloyannis}
\email{vkouk@aegean.gr}
\affiliation{Department of Mathematics, University of the Aegean, Karlovasi, 83200 Samos, Greece}
\begin{abstract}
Expanding upon our prior findings on the proximity of dynamics between integrable and non-integrable systems within the framework of nonlinear Schrödinger equations, we examine this phenomenon for the focusing Discrete Gross–Pitaevskii equation in comparison to the Ablowitz–Ladik lattice. The presence of the harmonic trap necessitates the study of the Ablowitz--Ladik lattice  in weighted spaces.  We establish estimates for the distance between solutions in the suitable metric, providing a comprehensive description of the potential evolution of this distance for general initial data. These results apply to a broad class of nonlinear Schrödinger models, including both discrete and partial differential equations. For the Discrete Gross--Pitaevskii equation, they guarantee the long-term persistence of small-amplitude bright solitons, driven by the analytical solution of the AL lattice, especially in the presence of a weak harmonic trap.  Numerical simulations confirm the theoretical predictions about the proximity of dynamics between the systems over long times. They also reveal that the soliton exhibits remarkable robustness, even as  the effects of the weak harmonic trap become increasingly significant, leading to the soliton’s curved orbit.
\end{abstract}	
\maketitle

\section{Introduction}
	The Gross--Pitaevskii equation (GPE) is a fundamental equation in the field of condensed matter physics, particularly in the study of Bose--Einstein condensates (BECs). It is a nonlinear Schr\"odinger (NLS) equation, including an external trapping (magnetic or optical) potential, and the nonlinearity is introduced by the interatomic interactions \cite{pitaevskii-stringari, Kevre_BEC, Kevre_BEC2, Peli}.  Moreover, the study of GPE  gives rise to significant nonlinear effects and states, such as solitons and vortices, which play a pivotal role in various contexts, spanning from optical pulses (temporal solitons) to stationary beams (spatial solitons) of nonlinear optics \cite{newell_moloney}, and to soliton matter waves in macroscopic quantum systems. 
	
Let us recall  a short description of the GPE equation in one-spatial dimension (1D). The 1D-model describes the dynamics  of an atomic BEC confined in a highly anisotropic harmonic trap, $V_{HT}$, with frequencies $\omega_x$ and $\omega_{\perp} \equiv \omega_y=\omega_z$, such that $\omega_x \ll \omega_{\perp}$. Then, in the mean-field approximation, and for sufficiently low temperatures (so that thermal and quantum fluctuations can be considered negligible ),  \cite{Kevre_BEC}, \cite{Kevre_BEC2}, the BEC dynamics is described by the 1D, GPE equation:
\begin{eqnarray}
	i\hbar\frac{\partial\Psi}{\partial t}=-\frac{\hbar^2}{2m} \frac{\partial^2 \Psi}{\partial x^2}
	+ V_{ext}(x)\Psi +g_{1D}|\Psi|^2\Psi,\;\;x\in\mathbb{R}.
	\label{gpe}
\end{eqnarray}
In equation \eqref{gpe},  $\Psi(x, t)$ is the macroscopic BEC wavefunction normalized to the number of atoms, namely $\int |\Psi|^2  dx = N$ and $m$ is the atomic mass. The constant  $g_{\rm 1D}$ is the effectively 1D coupling constant which is proportional to the atomic scattering length. At this point, it is important to notice that when $g_{\rm 1D}>0$ the interatomic interactions are repulsive and the model is defocusing (repulsive). When  $g_{\rm 1D}<0$ the interatomic interactions are attractive, and the model is focusing.  The external potential, $V_{ext}(x)$, in Eq.~(\ref{gpe}) takes the form
\begin{eqnarray}
	V_{ext}(x)\equiv V_{HT}(x) = \frac{1}{2}m \omega_x^2 x^2.
	\label{Vext}
\end{eqnarray}
Dimensionless GPE  equations may have the form:
\begin{equation}
	i \frac{\partial \phi}{\partial t} = - \frac{\partial^2 \phi}{\partial x^2}+ \Omega^2 x^2\phi
	+ g|\phi|^{2}\phi,\;\;x\in\mathbb{R}.
	\label{dim1dgpe}
\end{equation}
The parameter $\Omega$ is the {\em strength of the harmonic trap}. In the  context of nonlinear optics \cite{nlox2}, \cite{HorNI}, $\phi$ is the normalized electric field envelope, $t$ denotes the propagation direction. The parameter $\Omega$ accounts for the change in the refractive index of the medium in the $x$-direction (transverse to the propagation). Equation \eqref{dim1dgpe} motivates us to consider in this work, its spatial discretization, called as the discrete GPE equation (DGPE): 
	\begin{equation}\label{GPE}
			i \dot{\phi}_n + \frac{1}{h^2} (\phi_{n+1} - 2\phi_n + \phi_{n-1})  - \Omega^2(hn)^2\phi_n -g |\phi_n |^2 \phi_{n} = 0,\,\,\,n\in {\mathbb{Z}},
	\end{equation}
where $h$ denotes the 1D-lattice spacing. The relevance of the DGPE equation in the context of BECs, through the so called tight-binding approximation,  is discussed in detail in \cite{JMP2011} and the references therein. For other numerous applications in distinct physical contexts we refer to \cite{KevreDNLS}, \cite{BorisAll}.

The primary motivation for studying the DGPE equation \eqref{GPE} stems from the following well-known observation: In the limit of small external potentials   $\Omega\rightarrow 0$, \cite{Kevre_BEC2}, \cite{Peli}, the GPE partial differential equation (pde) becomes {\em the completely integrable NLS equation. However, when $\Omega\rightarrow 0$,  the DGPE transforms into the non-integrable DNLS equation. This represents a significant difference between the GPE pde and its DGPE nonlinear lattice counterpart.} Nevertheless, in the realm of nonlinear lattices, the Ablowitz--Ladik (AL) equation 
	\begin{equation}\label{AL} 
		i \dot{\psi}_n + \frac{1}{h^2} (\psi_{n+1} - 2\psi_n + \psi_{n-1}) + \mu | \psi_n|^2 (\psi_{n+1} + \psi_{n-1}) = 0,\,\,\,n\in {\mathbb{Z}},
	\end{equation}
is the famous fully integrable system.  The DGPE equation \eqref{GPE}, as another significant model, is particularly suited for investigation in line with our previous work \cite{DJNa,DJNb,JDE2024}, which explores the potential persistence of integrable wave dynamics from integrable Hamiltonian systems to non-integrable ones. By once again using the AL equation as the fundamental integrable model, we establish proximity estimates between the solutions of the DGPE equation and the AL lattice.

We restrict the present study to  the {\em the focusing} case  of DGPE ($g<0$), supplementing both systems with zero boundary conditions at infinity 
\begin{equation}
	\label{Vbc}
	\lim_{|n|\rightarrow\infty} \phi_n(t)=	\lim_{|n|\rightarrow\infty} \psi_n(t) =0,\;\;\mbox{for all $t\geq 0$}.
\end{equation}
We note that the defocusing case $(g>0$) is of particular interest and physical significance due to its relevance in the study of discrete dark solitons, which play a prominent role in the dynamics of BECs \cite{Kevre_BEC2} and in nonlinear optics \cite{BorisAll}. Investigating such solutions necessitates supplementing the DGPE system with non-zero boundary conditions at infinity. This represents a fundamentally different boundary value problem compared to the zero boundary conditions \eqref{Vbc}, as highlighted in \cite{JNLS2024}. The associated issues for the DGPE are  addressed in an ongoing work. Thus, in the
present work, we will consider the bright soliton dynamics induced by the analytical solution of the focusing AL lattice \eqref{AL} ($\mu>0$) as the major example to explore the potential persistence of integrable dynamics in the DGPE \eqref{GPE}. Specifically, we consider the bright soliton given by
\begin{align}
	\label{ALBS}
	\psi_n(t) &= A\mathrm{sech}\big(b(x_n - v t)\big) e^{i(a x_n - \omega t)},\;\;\;A=\frac{\sinh(bh)}{h\sqrt{\mu}},\\
	\omega &=-2 \frac{\cos(ah)\cosh(bh)-1}{h^2},\nonumber\\
	v &=\frac{2\sin(ah)\sinh(bh)}{bh^2},\;\;x_n=hn.\nonumber
\end{align}
We note, however, that the proximity estimates between the solutions of the compared focusing systems are general and apply to any solutions that satisfy the boundary conditions \eqref{Vbc}. To derive these estimates (following the approach in \cite{DJNa}, \cite{DJNb}), additional novel techniques are introduced to enrich their proofs, due to the presence of the external potential in the DGPE. According to the boundary conditions \eqref{Vbc} the natural phase spaces for studying the proximity of the dynamics are the standard sequence spaces $l^p$, $1\leq p\leq\infty$.  However, the approach requires the introduction of {\em weighted spaces} introduced by harmonic potential, and estimates for the solutions of the AL lattice in these weighted spaces. Despite the system's integrability, addressing the AL system in weighted sequence spaces poses significant challenges, as discussed in \cite{Yamane1, China1} for the focusing case and \cite{Yamane2, China2} for the defocusing case. More broadly, analyzing nonlinear lattices in weighted spaces is crucial, as it can establish the existence of solutions with a prescribed rate of spatial localization. Therefore, we present a general global existence result for the Salerno lattice \cite{Salerno}, \cite{Cai94}, \cite{R3}
\begin{equation}
	i\dot{\psi_n}+ \frac{1}{h^2} (\psi_{n+1} - 2\psi_n + \psi_{n-1})+\mu|\psi_{n}|^2)(\psi_{n+1} +\psi_{n-1})+
	\gamma|\psi_n|^2\psi_n=0,\,\,\,n\in {\mathbb{Z}},\label{eq:Salerno}
\end{equation}
which  encompasses both the DNLS (for $\mu=0$) and AL (for $\gamma=0$) systems  and we will comment explicitly for the  AL  lattice which is our primary concern in our applications.

The present study has motivated a comprehensive description of the potential long-term behavior of the distance between the solutions of the DGPE and the AL lattice. The theoretical results establish a general framework applicable to a broad class of NLS-type systems. As demonstrated in \cite{DJNa, DJNb, JDE2024}, this framework addresses the fundamental question of stability \cite{book2}, comparing the dynamics of an integrable system with those of its non-integrable counterparts. The distance always evolves within a trapping region defined by the graphs of functions that describe the time-dependent upper estimates for its transient growth, horizontal lines corresponding to its global, uniform-in-time bounds, and is unbounded only in the time direction. This characteristic, which will be discussed in detail in Section \ref{SecIIIB}, is typical of systems whose solutions remain uniformly bounded in time by quantities that depend only on their their initial data. It follows as a consequence of the triangle inequality when applied to the distance metric.  The estimates, which depend on the initial data and on $\Omega$, may define narrow regions when the initial data and $\Omega$ are small. Small variations in the distance within these regions suggest a potential proximity between the dynamics of the non-integrable system and the integrable one. This theoretical result explains the long-time persistence of small-amplitude AL solitons in the dynamics of the DGPE system,  for small strengths $\Omega$ of the harmonic trap. Numerical findings further support this persistence, showing that the deviation of the distance within these narrow regions is significantly smaller than theoretically predicted. This finding is consistent with results from our previous works \cite{DJNa, DJNb, JDE2024}. 
\textcolor{blue}{}

Additionally, it is noteworthy that the persisting soliton exhibits remarkably robust evolution, even as its orbit becomes curved due to the effects of the weak harmonic trap. In the continuous limit \cite{DJFNon}, this effect can be explained through the derivation of equations of motion for the center of the GPE soliton. However, in the discrete case, deriving such equations may not be as straightforward. At this point, we should remark that although we are not aware of  works which specifically addressed the stability of the bright soliton solution \eqref{ALBS} in the AL lattice  in a rigorous form, the analysis of \cite[Chapter 2]{KevreDNLS} (which also relies on the integrability of the model and the translational invariance of \eqref{ALBS}), the results of \cite{ABH} and \cite[Chapter 4, Sec. 4.2]{BMB} (discussing the Vakhitov–Kolokolov criterion for the Salerno lattice) suggest that it is linearly and nonlinearly (orbitally) stable. An interesting direction for further study is to prove, following the transitivity arguments of \cite[Theorem 2.3]{JDE2024}, that the observed evolution of the AL soliton \eqref{ALBS} in the DGPE remains robust in the regime of proximal dynamics.

The paper is structured as follows. In  Section~\ref{SecII}, we discuss the functional set-up of the problem giving an emphasis in the weighted sequence spaces and the estimates of the solutions of the AL lattice in these spaces. These estimates is an essential ingredient for the proof of the proximity estimates between the solutions of the DGPE and the AL lattice, given in Section~\ref{SecIII}. Section \ref{SecIV} is devoted to the numerical results.   Finally, Section~\ref{SecV}, summarizes our findings with some further remarks on ongoing and future works.
\section{Functional set up and estimates for the Salerno lattice in weighted sequence spaces}
The section is divided in two parts. In the first part we discuss the sequence spaces required of the study of the initial-boundary value problems of the DGPE and the AL lattices when supplemented with the vanishing boundary conditions \eqref{Vbc}. In the second part we prove a global existence result for the solutions of the Salerno lattice \eqref{eq:Salerno} with initial data in $l^2_w$ and we comment specifically on the AL lattice under the lights of the results of \cite{China1}.
\label{SecII}

\subsection{Functional set up: standard and weighted sequence spaces and the discrete Laplacian in weighted spaces}
Due to the boundary conditions \eqref{Vbc}, the functional setting for the DGPE and AL initial-boundary value problems consists of the standard  sequence spaces 
\begin{equation}
	l^p=\Big\{ u=(u_n)_{\nZ}\,\in {\mathbb{C}},\quad \| u\|_{l^p}^p=h\sum_{\nZ}|u_n|^p\Big\}.
\end{equation}
The Hilbert space $l^2$ is endowed with the inner product  
$$
(u,v)_{l^2}=h\mathrm{Re}\sum_{\nZ}u_n\overline{v}_n,\quad u,\,v\in l^2.
$$
The continuous embeddings
\begin{equation}
	l^r\subset l^s,\quad \|w\|_{l^s}\le \|w\|_{l^r}, \quad1 \le r\le s \le \infty,\label{ineq_l}
\end{equation}
will be thoroughly used in the sequel. Let us also recall some basic properties of the linear difference operators involved in the models. For any $u \in l^2$, the 1D-discrete Laplacian $\Delta_d:l^2\to l^2$,
\begin{equation}
	\label{dL}
	(\Delta_d u)_{\nZ}=u_{n+1}-2u_n+u_{n-1},
\end{equation}
is bounded, self-adjoint on $D(\Delta_d)=l^2$ and $\Delta_d\le 0$, due to the relations
\begin{align*}	\label{prop}
	(\Delta_{d} u,v)_{l^2}&=(u,\Delta_{d}v)_{l^2}, \quad u,v\in l^2,\\
	(\Delta_{d}u,u)_{l^2}&=-h\sum_{\nZ}|u_{n+1}-u_n|^2\leq 0,\\
	\|\Delta_d u\|_{l^2}^2&\leq c\|u\|_{l^2}^2,\;\;c>0.
\end{align*}

We will also consider the weighted sequence space $l_w^2$
\begin{align*}
	l^2_w  := \Big\{u_n\in \mathbb{C}, \quad
	\|u\|_{l^2_w}^2 = h\sum_\nZ{w_n|u_n|^2} < \infty \Big\},
\end{align*}
where the weight function $w_n$ satisfies the following conditions

\begin{eqnarray}\label{w_cond}
	w_n &\ge& 1, \nonumber \\
	 w_n &\le& c_1 w_{n+1}, \quad {c}_1>0.
\end{eqnarray}
The space $l^2_w$ is a Hilbert space with inner product
$$
(u,v)_{l^2_w}=h\mathrm{Re}\sum_{\nZ}w_nu_n\overline{v}_n,\quad u,\,v\in l^2_w.
$$
\begin{rem}\label{rem2}
	Standard simple examples of weights that satisfy the conditions \eqref{w_cond} with $c_1=1$, are the following: (a) For $\lambda, k\in\mathbb{N}$,   $w_n=w_{n, alg}= (1+|hn|^\lambda)^k$, for  $\lambda>0$, $k\geq 1$, which implies at least an algebraic localization of the sequences of $l^2_w$. (b) $w_n=w_{n,exp}= e^{\lambda|n|^k}$, for $\lambda>0$, $k\geq 1$, which implies exponential localization of the sequences of $l^2_w$. 
\end{rem}
The discrete Laplacian operator \eqref{dL} is not symmetric in the weighted inner product of $l^2_w$. However, it is still a bounded operator on $l^2_w$  as it is proved in the following proposition.
\begin{prop} 
	\label{DLW}
	We assume that the weight $w_n$ satisfies the  conditions \eqref{w_cond}. Then, 
	the discrete Laplacian \eqref{dL} defines a bounded operator on  $l^2_w$, i.e., 
	$$
	\Delta_d: l^2_w\to l^2_w,
	$$ 
	and
	$$
	\|\Delta_d(u)\|_{l^2_w} \le 2\sqrt{1+c_1+2c_1^{1/2}}\|u\|_{l^2_w}. 
	$$ 
\end{prop}	
\begin{proof} For any $u\in \ell^2_w$, by using  the Cauchy--Schwarz inequality and assumptions \eqref{w_cond} on the weight $w_n$, we have that
	\begin{align*}
		\|\Delta_d u\|^2_{l^2_w} = & h \sum_{\nZ}w_n|u_{n+1} -2 u_n  + u_{n-1}|^2 \\ 
		\le & h\sum_{\nZ}w_n\big(|u_{n+1}|^2 + 4|u_n|^2 + |u_{n-1}|^2 +   4|u_{n+1}| |u_n| + 4|u_{n-1}| |u_n|  + 2|u_{n+1}| |u_{n-1}|\big) \\ = &
		h\sum_{\nZ}w_n|u_{n+1}|^2 +4h\sum_{\nZ}w_n|u_n|^2 + 	h\sum_{\nZ}w_n|u_{n-1}|^2 + \\ &
		4h\sum_{\nZ}w_n|u_{n+1}| |u_n| + 4h\sum_{\nZ}w_n|u_{n-1}| |u_n|  + 2h\sum_{\nZ}w_n|u_{n+1}| |u_{n-1}|  \\ \le & 
		c_1h\sum_{\nZ}w_{n+1}|u_{n+1}|^2  + 4\|u\|^2_{l^2_w} + h\sum_{\nZ}w_{n-1}|u_n|^2  + 	4 \Big(h\sum_{\nZ}w_n|u_{n+1}|^2\Big)^{1/2} \Big(h\sum_{\nZ}w_n|u_n|^2\Big)^{1/2}  \\ &
		+ 4 \Big(h\sum_{\nZ}w_n|u_{n-1}|^2\Big)^{1/2} \Big(h\sum_{\nZ}w_n|u_n|^2\Big)^{1/2}   + 2 \Big(h\sum_{\nZ}w_n|u_{n+1}|^2\Big)^{1/2} \Big(h\sum_{\nZ}w_n|u_{n-1}|^2\Big)^{1/2}   \\ & \le
		c_1\|u\|^2_{l^2_w} + 4\|u\|^2_{l^2_w} + c_1h\sum_{\nZ}w_n|u_n|^2  +  4c_1^{1/2}\|u\|^2_{l^2_w} + 4c_1^{1/2}\|u\|^2_{l^2_w} + 2c_1\|u\|^2_{l^2_w} \\ & =
		4(1+c_1 + 2c_1^{1/2})\|u\|^2_{l^2_w}.
	\end{align*}	
\end{proof}
We will also use the continuous embedding
\begin{equation}
\label{wemb}
\|u\|_{l^2}\leq \|u\|_{l^2_w},\;\;\mbox{for all $u$ in $l^2_w$},
\end{equation}
which obviously holds due to the first of the properties \eqref{w_cond} on the weight $w_n$.
\subsection{Global existence of solutions for the Salerno lattice on $l^2_w$ and remarks on the solutions of the AL lattice on $l^2_w$}
We set $\gamma,\mu>0$ in the Salerno lattice \eqref{eq:Salerno}. It is an infinite dimensional, non-integrable Hamiltonian system with Hamiltonian
\begin{eqnarray}
	\label{Sham}
	H_{\small{S}}= \sum_\nZ \bar{\psi}_n(\psi_{n+1} -  \psi_{n-1})-\frac{\gamma h^2}{2}\sum_\nZ|\psi_n|^2-\frac{2}{\mu h^2}\sum_\nZ\ln(1+\mu|\psi_n|^2).
\end{eqnarray}
When supplemented with the vanishing boundary conditions  the corresponding initial-boundary value problem is globally well-posed in $l^2$. 
\begin{prop}
	\label{SGB}
Let $\mu,\gamma\geq 0$.	For every initial condition $\psi_n(0):=\psi(0)\in l^2$,
	the problem (\ref{Vbc})-(\ref{eq:Salerno}) possesses a unique global solution $\psi_n(t):=\psi(t)$ belonging to $C^1([0,\infty),l^2)$.  Moreover, the solution is uniformly bounded in $l^2$ for all $t> 0$, and satisfies the estimate 
	\begin{equation}
		\label{UBS}
			\|\psi(t)\|_{l^2}^2\leq M_{\mu,\gamma}\|\psi(0)\|_{l^2}^2,\;\;\mbox{for all $t>0$},
	\end{equation}
for some constant $M_{\mu,\gamma}>0$, which is independent of $t$.
\end{prop}
The proof is given in \cite{DJNa}. Proposition \ref{SGB} provides simultaneously the global-in-time, well-posedness of the corresponding problems for the focusing non-integrable DNLS and the  focusing integrable AL lattice. We will use Proposition \ref{SGB} to prove the following theorem.
\begin{theorem}\label{THS} We assume that the weight function $w_n$ satisfies the conditions \eqref{w_cond}. 
Let $\psi_n(0)\in l^2_w$. Then, the corresponding solution of the initial-boundary value problem  for the Salerno lattice \eqref{Vbc}-\eqref{eq:Salerno} $\psi(t)\in C^1([0,\infty),l^2_w)$.
\end{theorem}
\begin{proof} The assumption that the initial condition $\psi(0)\in l^2_w$ implies due to the embedding \eqref{wemb}, that $\psi(0)\in l^2$. Therefore, Proposition \ref{SGB} is applicable and $\psi(t)\in C^1([0,\infty),l^2)$. 
Multiplying \eqref{eq:Salerno} with $hw_n\overline{\psi}_n$, taking the sum over $n$ and keeping the imaginary parts, yields
\begin{equation}\label{eq_lem2}
	\frac{1}{2}\frac{d}{dt}h\sum_\nZ{w_n|\psi_n|^2} +  \frac{1}{h^2} \mathrm{Im}\Big\{h\sum_\nZ{w_n\overline{\psi}_n\Delta_d \psi_n}\Big\} +   \mu\,\mathrm{Im}\Big\{h\sum_\nZ{w_n\overline{\psi}_n|\psi_n|^2}(\psi_{n+1} + \psi_{n-1})\Big\} = 0.
\end{equation}
We observe that with the above multiplication, the nonlinear term stemming from the cubic nonlinearity $\gamma|\psi|^2\psi_n$ of the Salerno lattice \eqref{eq:Salerno} is eliminated. Next, by using Proposition \ref{DLW} and the Cauchy-Schwarz inequality, the second term of \eqref{eq_lem2}, can be estimated as follows:
\begin{align*}
	\Big|\im\Big\{h\sum_\nZ{w_n\overline{\psi}_n\Delta_d \psi_n}\Big\}\Big| & \le 
	h\sum_\nZ{w_n^{1/2}|\overline{\psi}_n| w_n^{1/2}|\Delta_d \psi_n|} \\ & \le
	\Big(h\sum_\nZ w_n|\psi_n|^2 \Big)^{1/2}	\Big(h\sum_\nZ{w_n|\Delta_d\psi_n|^2 } \Big)^{1/2} \\ & =
	\|\psi\|_{l^2_w}	\|\Delta_d\psi\|_{l^2_w} \\ & \le
	2\sqrt{1 + c_1 + 2c_1^{1/2}}\|\psi\|^2_{l^2_w}.
\end{align*}

Finally, for the last  term of \eqref{eq_lem2}, we have the estimate 
\begin{align*}
	%\label{Term1}
	\Big|\mathrm{Im} \Big\{h \sum_\nZ  |\psi_n |^2 (\psi_{n+1} + \psi_{n-1})  w_n\overline{\psi}_n \Big\} \Big|&  
	\le 2 \sup_{n}| \psi_n | h\sum_\nZ w_n | \psi_n |^2 |\psi_n | \\ & \le 
	2  \| \psi \|_{l^ \infty}^2 h\sum_\nZ w_n  |\psi_n |^2 \\ & \le
	2  \| \psi \|_{l^ \infty}^2 \| \psi \|^2_{l^2_w} \\ & \le
	2\| \psi \|_{l^2}^2 \| \psi \|_{l^2_w}^2 \\ & \le
	2M_{\mu,\gamma}\| \psi(0) \|_{l^2}^2 \| \psi \|_{l^2_w}^2,
\end{align*}
where we used \eqref{ineq_l} and the uniform bound \eqref{UBS} for the solutions of the problem in $l^2$. 
Inserting the above estimates in \eqref{eq_lem2}, we get that 
$$
\frac{1}{2}\frac{d}{dt}  \| \psi \|_{l^2_w}^2  \le 
\Gamma \| \psi \|_{l^2_w}^2, \qquad \Gamma :=   \frac{2\sqrt{1 + c_1 + 2c_1^{1/2}}}{h^2} + 2\mu M_{\mu,\gamma}\| \psi(0) \|_{l^2}^2, 
$$
which implies the estimate
\begin{align}
	\label{expest}
	\| \psi(t) \|_{l^2_w}\le 
	e^{\Gamma t}\| \psi(0) \|_{l^2_w},
\end{align}
which excludes the possibility of any time-singularity in $l^2_w$ and guarantees global existence in $l^2_w$ by the standard continuation arguments for ordinary differential equations \cite{zei85a}. 
\end{proof}
An intuitive example illustrating the relevance of Theorem \ref{THS} and the estimate \eqref{expest} is provided by the $l^2_w = l^{2,2}$-norm of the soliton solution \eqref{ALBS}. This norm [continuous (blue) curves], as depicted in Figure \ref{FigEX}, grows exponentially over time, without leading to a finite-time blow-up. The growth is driven by the exponential term in $t$, which results in exponential growth when $v > 0$. On the other hand, the $l^2$-norm [dashed (blue) horizontal lines] is conserved. Note that for the soliton solution \eqref{ALBS} the $l^2$-norm is not only bounded but also conserved,  while in general, the $l^2$-norm  is not conserved by the solutions of the AL, see  Section \ref{s3}.
%%%%%%%%%%%%%%%%%%%%%%%%%%%%555
%%%%%%%%%%%%%%%%%%%%%%%%%%%%%%%%%%%%%%%%%%%%%%%55555
\begin{figure}[tbp!]
	\centering 
	\begin{tabular}{cc}
		%		(a)&\hspace{0.7cm} (b)\\
		%		\includegraphics[scale=0.5]{norm_w_A_0_05_2.png}&
		%		\hspace{0.3cm}
		%		\includegraphics[scale=0.5]{norm_w_A_0_2_2.png}\\
		(a)&\hspace{0.7cm} (b)\\
		\includegraphics[scale=0.47]{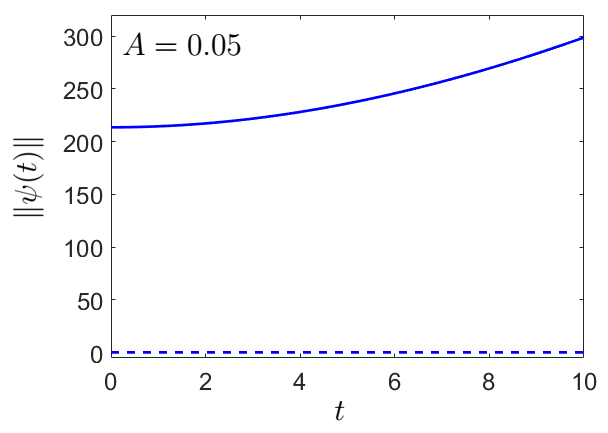}&
		\hspace{0.3cm}
		\includegraphics[scale=0.47]{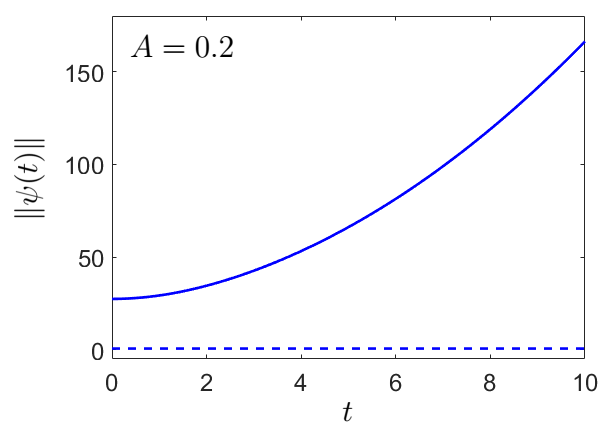}
	\end{tabular}
	\caption{
		\label{FigEX} An intuitive example relevant to the estimate \eqref{expest}. The $l^2_w$-norm (continuous blue curve) of the analytical soliton solution \eqref{ALBS} with $v>0$ and its corresponding $l^2$-norm (dashed blue curve). (a) $A=0.05$, $\|\psi\|_{l^2}=0.6345$. (b) $A=0.2$, $\|\psi\|_{l^2}=0.3163$.}
\end{figure}
%
%%%%%%%%%%%%%%%%%%%%%%%%%%%%%%%%
\begin{rem}
	\label{DPDE}
	\begin{enumerate}
		\item
		%\label{DPDE1}
{\em Continuous limit and the estimate \eqref{expest}.} We point out that the estimate \eqref{expest} should be considered only for {\em arbitrary, but finite values} of $h$. It is not possible to extend this understanding to the behavior of solutions in the weighted $L^2$ function spaces for the relevant pde (of NLS type, combining the standard power and derivative nonlinearity), which can be formally obtained when $h \to 0$. This limitation arises because the derivation of \eqref{expest} relies on the embedding \eqref{ineq_l}, which is valid only in the discrete setup. In the continuous case, deriving a similar estimate would require well-posedness and higher-order regularity results for the solutions, allowing the application of suitable Sobolev space embeddings on the infinite line. Although the study of such a problem is important, it is beyond the scope of the present work.
%\begin{figure}[tbp!]
%	\centering 
%	\begin{tabular}{cc}
%		(a)&\hspace{0.7cm} (b)\\
%		\includegraphics[scale=0.5]{norms_A_0_05_2.png}&
%		\hspace{0.3cm}
%		\includegraphics[scale=0.5]{norms_A_0_2_2.png}
%	\end{tabular}
%	\caption{
%		\label{FigEX2} \color{red} \bf An intuitive example relevant to the estimate \eqref{expest}. The $l^2_w$-norm of the analytical soliton solution \eqref{ALBS} with $v>0$ is depicted with the solid line. With the dashed line the corresponding constant $l^2$-norm is depicted (a) $A=0.05$ $\|\psi\|_{l^2}=0.6345$. (b) $A=0.2$  }
%\end{figure}
%%%%%%%%%%%%%%%%%%%%%%%%%%%%%%%%%%%%%%%%%%%%%%%%% 
%\paragraph{Remarks on Theorem \ref{THS}}
% \begin{enumerate}
	\item 
	%\label{DPDE2}
	{\em Solutions preserve the spatial localization of the initial condition.} Theorem \ref{THS} implies that the subspace $\mathcal{Y} = l^2_w$ of the space $\mathcal{X} = l^2$ is invariant under the flow of the Salerno lattice, $\mathcal{S}(t): \mathcal{X} \rightarrow \mathcal{X}$. In other words, $\mathcal{S}(t): \mathcal{Y} \subset \mathcal{X} \rightarrow \mathcal{Y} \subset \mathcal{X}$, meaning that for any $\psi(0) \in \mathcal{Y}$, it holds that $\psi(t) = \mathcal{S}(t)\psi(0) \in \mathcal{Y}$ for all $t > 0$. This invariance has physical significance: since $w_n$ describes a prescribed rate of localization, {\em the invariance implies that the solution preserves the localization of the initial condition}. In other words, initial conditions with algebraic or exponential spatial localization (see Remark \ref{rem2}) result in solutions with the same rate of spatial localization. This is important from both a mathematical and physical perspective, as localized solutions of physical significance are distinguished by the rate of their spatial localization.
	%%%%%%%%%%%%%%%%%
	\item 
	%\label{DPDE3}
	{\em The AL lattice in $l^2_w$.} Proposition \ref{SGB} generalizes \cite[Proposition 1]{Yamane1} for the AL lattice to the case of the Salerno lattice. The proof of \cite{Yamane1} makes use of the conservation law $c_{-\infty}=\Pi_{n\in\mathbb{Z}}(1+|\psi_n|^2)$ of the AL lattice, while our version for the Salerno lattice interpolates the global existence result of Proposition \ref{SGB} in $l^2$.  The problem of considering the AL lattice in weighted spaces is significant. The key works of \cite{Yamane1, China1} employ advanced methods from integrability theory to study the initial-boundary value problem in these weighted spaces, particularly for the case $\lambda=2$, $k\geq 1$, denoted by $l^2_w:=l^{2,k}$, {see Remark \ref{rem2}}. The choice of $k$ defines regularity properties of the reflection coefficient of the associated Riemann--Hilbert problem (RH). The study of the RH problem provides a detailed description of the global existence of solutions in $l^{2,k}$ and their long-time asymptotics in the $(n,t)$ plane. In summary, for initial data $\psi(0)\in l^2_w$, the solution is characterized as a sum of 1-bright soliton solutions of the form \eqref{ALBS} (with different phase shift formulas in different regions), while along specific rays away from the solitons, the solution exhibits decaying oscillations at particular rates. To the best of our knowledge, it is not known if solutions of the AL lattice with initial data in $l^2_w$ are uniformly bounded with respect to time, in $l^2_w$.
	\end{enumerate}
\end{rem}
%%%%%%%%%%%%%%%%%%%%%%%%%%%%%%%%
\section{The distance between the solutions of the AL and the DGPE lattice}
\label{SecIII}\label{s3}
In this section, we present the main result concerning the estimates for the distance between the solutions of the DGPE and the AL lattice. In the first part, we provide the estimates for general initial data. We comment also for the case of the finite lattice with Dirichlet boundary conditions.  In the second part, we discuss the application of these estimates for small initial data.
\subsection{The distance between the solutions: general initial data}
\label{SubsecIIa}
\paragraph{The case of the infinite lattice.} Recall  (see \cite{JMP2011}), that the non-integrable DGPE latttice \eqref{GPE} conserves the following quantities, the Hamiltonian 
\begin{equation}\label{H_DGPE} 
	H_{NI} =  \sum_{\nZ}|\phi_{n+1} -\phi_n|^2 + \Omega^2 h^2\sum_{\nZ}(hn)^2|\phi_n|^2 + \frac{g h^2}{2}\sum_{\nZ}|\phi_n|^4,
\end{equation}
and the power (or $l^2$-norm)
\begin{equation}
	\label{powerDGPE}
 P=\|\phi\|_{l^2}^2 = h\sum_{\nZ} |\phi_n|^2.	
\end{equation}

%%%%%%%%%%%%%%%%%%%%%%%%%%%%%%%%%%%
The integrable AL lattice \eqref{AL} has infinitely many invariant quantities \cite{PM1}, including the Hamiltonian
\begin{equation}\label{H2}
	H_{I} = \sum_\nZ \overline{\psi}_n(\psi_{n+1} -  \psi_{n-1}) - \frac{2}{\mu h^2} \sum_{\nZ} \ln(1+\mu h^2|\psi_n|^2) 
\end{equation}
and the deformed power 
\begin{equation}\label{norm}
	P_\mu= \frac{1}{\mu h^2}\sum_\nZ{\ln(1+\mu h^2|\psi_n|^2)}.
\end{equation} 
Therefore, for all $t>0$, we have that
\begin{eqnarray}
	\label{NGPE}
	H_{NI}(t) &=&H_{NI}(0), \quad H_{I}(t)  = H_{I}(0),\\
	\label{N}
	\|\phi(t)\|^2_{l^2}  &=& 	\|\phi(0)\|^2_{l^2}, \quad P_\mu(t) = P_\mu(0),
\end{eqnarray}
for the solutions of the DGPE and and the AL systems, respectively. Note that \cite{PM1} provides formulas for the derivation of the conserved quantities of general Hamiltonian DNLS-type systems.  For the AL lattice, Proposition \ref{SGB} implies a bound of the form \eqref{UBS}, when $\psi(0)\in l^2$, namely,
satisfies the estimate 
\begin{equation}
	\label{UBSAL}
	\|\psi(t)\|_{l^2}^2\leq M_{\mu,0}\|\psi(0)\|_{l^2}^2, \;\;\mbox{for all $t>0$}.
\end{equation}
For the conservation of the Hamiltonian \eqref{H_DGPE}, it is crucial to remark the following:  the global in time solvability in $l^2$, that is, that  $\phi(t)\in C^1([0,\infty),l^2)$ is guaranteed for all initial data $\phi(0)\in l^2$ due to the conservation of the power \eqref{powerDGPE}.  However, the fact that $\phi(t)\in C^1([0,\infty),l^2)$ do not guarantees that the Hamiltonian $H_{NI}(t)$ is  bounded,  unless 
\begin{eqnarray}
\label{ht0}
H_{NI}(0)<\infty.
\end{eqnarray}
For \eqref{ht0} to be satisfied, it is required that the initial data have {\em finite discrete variance}, that is 
 \begin{eqnarray}
 \label{ht01}
V_d(0):=h\sum_{\nZ} (hn)^2|\phi_n(0)|^2< \infty.
 \end{eqnarray}
Initial data satisfying \eqref{ht01} are in $l^2$, since 
\begin{eqnarray}
\label{ht001}
h^3\sum_{\nZ}|\phi_n(0)|^2=h^3|\phi_0(0)|+h^3\sum_{n\neq 0, n\in\mathbb{Z}}|\phi_n(0)|^2\leq h^3|\phi_0(0)|+V_d(0).
\end{eqnarray}
\begin{prop}
	\label{finHNI}
For every $\phi(0)$ satisfying \eqref{ht01}, let $\phi(t)\in C^1([0,\infty),l^2)$ be the global in time solution of the problem \eqref{GPE}-\eqref{Vbc} for the DGPE lattice. Then
\begin{equation}
\label{htrapterm}
V_d(t)=h\sum_{\nZ}(hn)^2|\phi_n(t)|^2<\infty,\;\;\mbox{for all $t>0$.}
\end{equation}
\end{prop}	
\begin{proof}
The fact that $\phi(t)\in C^1([0,\infty),l^2)$ follows directly from \eqref{ht001} and the conservation of the power  \eqref{powerDGPE}. Next, from the conservation of the Hamiltonian  	$H_{NI}(t) =H_{NI}(0)$, we have that
\begin{eqnarray}
\label{ht1}
 \Omega^2 h^2\sum_{\nZ}(hn)^2|\phi_n|^2 \leq  \sum_{\nZ}|\phi_{n+1} -\phi_n|^2 +\frac{|g| h^2}{2}\sum_{\nZ}|\phi_n|^4+\left|H_{NI}(0)\right|.
\end{eqnarray}
Since $\phi(t)\in C^1([0,\infty),l^2)$, the first term of the right-hand side of \eqref{ht1} is uniformly bounded for all $t>0$, due to the inequality
\begin{eqnarray}
\label{ht2}
\sum_{\nZ}|\phi_{n+1} -\phi_n|^2\leq 4\sum_{\nZ}|\phi_n|^2.
\end{eqnarray}
The second term of the right-hand side of \eqref{ht1} is also uniformly bounded for all $t>0$, due to the embedding $l^2\subset l^4$ (see \eqref{ineq_l}). Also, by the assumption \eqref{ht01} we have that the third term  $|H_{NI}(0)|<\infty$, and hence, \eqref{htrapterm} holds.
\end{proof}	
Obviously, Proposition \ref{finHNI} is valid for initial data  $\phi(0)\in\l^{2,2}$ of the DGPE. Note that the limitations on $h$ described in Remark \ref{DPDE} are also valid for the case of  Proposition \ref{finHNI}.

Before proceeding to the statement and proof of the main result, we summarize the well-posedness results for the DGPE and AL lattices. The problem \eqref{GPE}-\eqref{Vbc} for the DGPE lattice is globally well-posed in time in $l^2$, with uniformly bounded solutions
for all initial data $\phi(0)\in l^2$.
However, the Hamiltonian may be unbounded if the initial data do not decay faster than quadratically, as dictated by the finite variance condition \eqref{ht01}. The same uniform boundedness of solutions holds for the problem \eqref{AL}-\eqref{Vbc} for the AL lattice in $l^2$.  Since the analytical solution \eqref{ALBS} of the latter, which we aim to examine for its potential persistence in the DGPE dynamics, is bounded in $l^2$ but may have an $l^{2,2}$-norm that  becomes unbounded (without singularities) in time, the  global existence results above- ensuring $l^2$-uniformly bounded solutions, suggest that  the $l^2$-metric is more suitable and less cumbersome for our initial investigations of the potential proximal dynamics. However, as we will see below, the presence of the harmonic trap term in the equation governing the difference of solutions between the two systems will necessitate considering at least the problem \eqref{AL}-\eqref{Vbc} for the AL lattice in $l^{2,2}$. While this problem remains globally well-posed in $l^{2,2}$ it generally lacks uniformly bounded solutions. Measuring the distance in weighted spaces is a problem of interdependent interest and will be considered elsewhere. 
 
The main result is proved in the following theorem.

\begin{theorem}
	\label{MR}
Let $g<0$ and $\mu>0$. We assume that the initial condition for the DGPE lattice \eqref{GPE} $\phi(0)\in l^2$, and that the initial condition for the AL lattice \eqref{AL}  $\psi(0)\in l^2_w=l^{2,2}$. Then, the distance between  the solutions of the systems $\|y(t)\|_{l^2}=\|\phi(t)-\psi(t)\|_{l^2}$, which satisfy the boundary conditions \eqref{Vbc},  grows transiently, with  at most an exponential rate according to the estimate 
\begin{eqnarray}\label{final1}
	\| y(t) \|_ {l^2}  &\le& \| y(0) \|_ {l^2}  + \mathcal{A}_1t+\mathcal{A}_2\exp\left(\Gamma t\right),\\ 
	\label{final1a}
	\mathcal{A}_1 &=&|g|  \|\phi(0)\|_ {l^2}^3  + 
	2\mu C_{\mu,0}\| \psi(0) \|_ {l^2}^3,\;\;
	\mathcal{A}_2=\frac{\Omega^2}{\Gamma} \| \psi(0) \|_{l^2_w},
\end{eqnarray}
for some  constant $C_{\mu,0}>0$, which is independent of $t$ and the initial data. In particular, the growth of the upper bound holds for finite time
\begin{eqnarray}
\label{finit}
t\in [0, t^*],\;\;\mbox{satisfying the inequality}\;\;\mathcal{A}_1t&+&\mathcal{A}_2\exp\left(\Gamma t\right)\leq \mathcal{B}-\|y(0)\|_{l^2},\\
\label{finit2}
\mathcal{B}[\phi(0),\psi(0)]&=&\|\phi(0)\|_{l^2}+C^*_{\mu,0}\|\psi(0)\|_{l^2},
\end{eqnarray}
where  $C^*_{\mu,0}>0$ is again independent of $t$ and the initial data. The quantity $\mathcal{B}[\phi(0),\psi(0)]$ is the global, uniform-in-time bound for $\|y(t)\|_{l^2}$, that is, $\|y(t)\|\leq \mathcal{B}$, for all $t>0$. 
\end{theorem}
\begin{proof}
	We start by subtracting the AL lattice \eqref{AL} from the DGPE lattice \eqref{GPE} to get the evolution equation for the difference of the solutions of the systems $y_n=\phi_n-\psi_n$, which reads as 
	\begin{equation}\label{main0}
		i \dot{y}_n  +\frac{1}{h^2}  \Delta_d y_n =  \Omega^2(hn)^2\phi_n + g|\phi_n|^2 \phi_n  + \mu |\psi_n|^2 (\psi_{n+1} - \psi_{n-1}).
	\end{equation}
	We apply the trick of simultaneously adding and subtracting the term $\Omega^2(hn)^2\psi_n$ in \eqref{main0} to avoid working directly with the distance $y$ in a weighted norm. Instead, we handle only the solution of the AL lattice $\psi$ in such a norm, by using Proposition \ref{SGB}, as will become evident below. With this trick, the equation \eqref{main0} is rewritten as
	\begin{equation}\label{main}
		i \dot{y}_n  +\frac{1}{h^2}  \Delta_d y_n =  \Omega^2(hn)^2y_n + g|\phi_n|^2 \phi_n  + \mu |\psi_n|^2 (\psi_{n+1} - \psi_{n-1}) + \Omega^2(hn)^2\psi_n.
	\end{equation}
	Multiplying \eqref{main} by $ h\overline{y}_n $, taking the sum over $n$  and keeping the  imaginary parts, we get the evolution equation for the distance in the $l^2$-metric $\|y(t)\|_{l^2}=\|\phi(t)-\psi(t)\|_{l^2}$:
	\begin{equation}\label{imeq}
		\frac{1}{2} \frac{d}{dt} \|y\|_{l^2}^2  =
		g\,\mathrm{Im} \Big\{h\sum_\nZ | \phi_n |^2 \phi_n \overline{y}_n\Big\} + 
		\mu\,\mathrm{Im} \Big\{h\sum_\nZ |\psi_n |^2 (\psi_{n-1}  + \psi_{n+1})\overline{y}_n\Big\} +
		\Omega^2\,\mathrm{Im} \Big\{h\sum_\nZ (hn)^2 \psi_n \overline{y}_n\Big\},
	\end{equation}
	noting that
	$\mathrm{Im}\Big\{ \sum_\nZ n^2 |y_n|^2\Big\}= 0$.
	Using  the embedding \eqref{ineq_l} and the Cauchy--Schwarz inequality, the first  and the second term of the right-hand side can be estimated as follows:
	\begin{eqnarray}
		\label{distest1}
		\Big| \mathrm{Im} \Big\{h\sum_\nZ |\phi_n|^2 \phi_n \overline{y}_n\Big\}\Big| &\le&  \|\phi\|_ {l^2}^3 \|y\|_{l^2},\\
		\label{distest2}
		\Big|\, \mathrm{Im} \Big\{h\sum_\nZ |\psi_n|^2 (\psi_{n-1} + \psi_{n+1}) \overline{y}_n\Big\} \Big|
		&\le& 2  \| \psi \|_ {l^2} ^ 3 \| y\| _ {l^2},
	\end{eqnarray}
	respectively. For the third term, the estimate will require a bound for the solution $\psi(t)\in l^2_w$, for all $t>0$, with the weight  $w_n = (1 + (hn)^2)^2$, that is, the weighted space $l_w^2=l^{2,2}$.  This bound is provided by  \eqref{expest} of Theorem \ref{THS}, for $\gamma=0$ which corresponds to the case of the AL lattice. Thus, applying the Cauchy--Schwarz inequality, we have that
	\begin{align}
		\label{distest3}
		\Big|\mathrm{Im} \Big\{h\sum_\nZ (hn)^2 \psi_n \overline{y}_n\Big\}\Big| & \le 
		h\sum_\nZ (1 + (hn)^2) |\psi_n| |\overline{y}_n|\nonumber \\ & \le 
		\|\psi\|_{l^2_w} \|y\|_{l^2} \nonumber\\ & \le
		\exp\left(\Gamma t\right)\| \psi(0) \|_{l^2_w} \|y\|_{l^2}.
	\end{align}
	Inserting the estimates \eqref{distest1}-\eqref{distest3} into equation \eqref{imeq}, we derive the differential inequality
	\begin{eqnarray}
		\label{distest4}
		\frac{1}{2} \frac{d}{dt} \|y\|_{l^2}^2  \le  
		|g|  \|\phi\|_ {l^2}^3 \|y\|_{l^2} + 
		2\mu  \| \psi \|_ {l^2} ^ 3 \| y\| _ {l^2} +
		\Omega^2	\exp\left(\Gamma t\right)\| \psi(0) \|_{l^2_w} \|y\|_{l^2}.
	\end{eqnarray}
In the first and second terms on the right-hand side of \eqref{distest4}, we apply the conservation of the $l^2$-norm \eqref{N} for $\phi(t)$ and the bound \eqref{UBSAL} for $\psi(t)$, respectively. The latter is applicable because $l^2_w = l^{2,2}$ is continuously embedded in $l^2$, as stated in \eqref{wemb}. Therefore, the assumption that $\psi(0) \in l^2_w$ implies that $\psi(0) \in l^2$, and for the corresponding solution measured in the $l^2$-norm, the estimate \eqref{UBSAL} holds. Then, \eqref{distest4} becomes
	\begin{eqnarray}
		\label{distest4a}
		\frac{1}{2} \frac{d}{dt} \|y(t)\|_{l^2}^2  \le  
		|g|  \|\phi (0)\|_ {l^2}^3 \|y(t)\|_{l^2} &+& 
		2\mu C_{\mu,0} \| \psi(0) \|_ {l^2} ^ 3 \| y(t)\| _ {l^2}\nonumber \\
		&+&
		\Omega^2\exp\left(\Gamma t\right)\| \psi(0) \|_{l^2_w}	\|y(t)\|_{l^2},\;\;C_{\mu,0}=M_{\mu,0}^{3/2}.
	\end{eqnarray}
	With the help of the identity
	$$
	\frac{d}{dt}  \|y\| _{l^2} ^ 2 = 2 \| y \| _ {l^2} \frac{d}{dt} \| y \|_ {l^2},
	$$
	the inequality \eqref{distest4a} gets the form
	\begin{eqnarray}
		\label{ineqfin}
		 \frac{d}{dt} \|y(t)\|_{l^2}  \le  
		|g|  \|\phi (0)\|_ {l^2}^3  + 
		2\mu C_{\mu,0} \| \psi(0) \|_ {l^2} ^ 3  +
		\Omega^2\exp\left(\Gamma t\right)\| \psi(0) \|_{l^2_w}.
	\end{eqnarray}
	Integration of  \eqref{ineqfin} with respect to time, provides the claimed  estimate \eqref{final1}, with growth rates $\mathcal{A}_1$ and $\mathcal{A}_2$ given by \eqref{final1a}, for $\|y(t)\|_{l^2}$.

Note that the estimate \eqref{final1} is {\em general, with respect to the initial data and parameters}.   It implies {\em at most, exponential growth} of the distance $\|y(t)\|_{l^2}$ with rate $\mathcal{A}_2$, which depends on the initial data and parameters as given in \eqref{final1a}.
We show now that this exponential growth rate is {only transient}. Due to the triangle inequality, we have that the distance $\|y(t)\|_{l^2}$ is { \em globally bounded,  uniformly in time}, by the norms of the initial conditions:
\begin{eqnarray}
\label{trans1}
\|y(t)\|_{l^2}&\leq& \|\phi(t)\|_{l^2}+\|\psi(t)\|_{l^2} \nonumber\\ &\leq& \|\phi(0)\|_{l^2}+C^*_{\mu,0}\|\psi(0)\|_{l^2}:=\mathcal{B}[\phi(0),\psi(0)],\;\;\;C^*_{\mu,0}=M_{\mu,0}^{1/2}.
\end{eqnarray}
Therefore, on the account of \eqref{final1}-\eqref{final1a}, it holds that
\begin{eqnarray}
\label{trans2}
||y(t)||_{l^2}\leq \| y(0) \|_ {l^2}  + \mathcal{A}_1t+\mathcal{A}_2\exp\left(\Gamma t\right)\leq \mathcal{B},\;\;\mbox{for all $t\geq 0$},
\end{eqnarray}
since the distance can never exceed its global  bound \eqref{trans1}, and any of its other bounds is less or equal than $\mathcal{B}$, for all $t\geq 0$. Therefore, any growth of the distance is only transient, and holds at most for $t\in [0, t^*]$, satisfying \eqref{trans2}.   This way we proved \eqref{finit} and \eqref{finit2}, and thus, completed the proof of the claims of the theorem.
\end{proof}
\paragraph{The case of the finite lattice with Dirichlet boundary conditions.} The case of the finite lattice is important for the numerical simulations. We consider the case of the Dirichlet boundary conditions. Then, the systems are considered in the finite dimensional spaces
\begin{equation*}
	{l}^p_{0}:=\bigg\{U=(U_n)_{n\in\mathbb{Z}}\in\mathbb{C}:\quad U_0=U_{N}=0,\quad
	\|U\|_{\ell^p_{0}}:=\bigg(h\sum_{n=1}^{N-1}|U_n|^p\bigg)^{\frac{1}{p}}<\infty\bigg\}, \quad 1\leq p\leq\infty,
\end{equation*}
for a finite lattice of $N+1$ nodes. In the finite dimensional setting all norms are equivalent. Thus, considering the weighted norm $l^2_w=l^2_2$, there exist  positive constant $c_2(N)$, such that
\begin{eqnarray}
	\label{normequiv}
\|\phi\|_{l^2_0}\leq \|\phi\|_{l^2_w} \leq c_2\|\phi\|_{l^2_0},\;\; c_2(N)\sim 1+h^2N^2,\;\; \quad \text{for all}\quad  \phi\in l^2_{0}.
\end{eqnarray}
The equivalence of norms \eqref{normequiv} modifies the upper bounds of Theorem \ref{MR} {\em to be linear}, according to  the following corollary. 
\begin{corollary}
	\label{MRF}
	Let $g<0$ and $\mu>0$.  Consider the finite DGPE and AL lattices in $l^2_0$, that is, supplemented with Dirichlet boundary conditions, and  their initial conditions $\phi(0),\psi(0)\in l^2_0$.  Then, the distance between  the solutions of the systems $\|y(t)\|_{l^2_0}=\|\phi(t)-\psi(t)\|_{l^2_0}$  grows transiently, with  at most a linear rate according to the estimate 
	\begin{eqnarray}\label{final1D}
		\| y(t) \|_ {l^2_0}  &\le& \| y(0) \|_ {l^2_0}  + \mathcal{A}t,\\ 
		\label{final1aD}
		\mathcal{A} &=&|g|  \|\phi(0)\|_ {l^2_0}^3  + 
		2\mu \tilde{C}_{\mu,0}\| \psi(0) \|_ {l^2_0}^3 +\Omega^2\tilde{C}_{1,\mu,0}\|\psi(0)\|_{l^2_0},\;\;
	\end{eqnarray}
	for some  constants $\tilde{C}_{\mu,0}>0,\,\tilde{C}_{1,\mu,0} $, which are independent of $t$ and the initial data. In particular, the growth of the upper bound holds for finite time
	\begin{eqnarray}
		\label{finitD}
		t\in [0, t^*],\;\; t^*&=&\frac{\mathcal{B}-\|y(0)\|_{l^2_0}}{\mathcal{A}},\\
		\label{finit2D}
		\mathcal{B}[\phi(0),\psi(0)]&=&\|\phi(0)\|_{l^2_0}+\tilde{C}^*_{1,\mu,0}\|\psi(0)\|_{l^2_0},
	\end{eqnarray}
	where  $\tilde{C}^*_{1,\mu,0}>0$ is again independent of $t$ and the initial data. The quantity $\mathcal{B}[\phi(0),\psi(0)]$ is the global, uniform-in-time bound for $\|y(t)\|_{l^2_0}$, that is, $\|y(t)\|_{l^2_0}\leq \mathcal{B}$, for all $t>0$. 
\end{corollary}
 \begin{proof}
The equivalence of norms \eqref{normequiv} modifies the estimate \eqref{distest3} for the AL lattice. The bound \eqref{UBSAL} holds also for the finite lattice, that is,
\begin{equation}
	\label{UBSALD}
	\|\psi(t)\|_{l^2_0}^2\leq M_{\mu,0,f}\|\psi(0)\|_{l^2_0}^2,\;\;\mbox{for all $t>0$}.
\end{equation}
Then, instead of the inequality \eqref{distest3}, we have its counterpart
\begin{align}
	\label{distest3D}
	\Big|\mathrm{Im} \Big\{h\sum_{n=1}^{N-1} (hn)^2 \psi_n \overline{y}_n\Big\}\Big| & \le 
	h\sum_{n=1}^{N-1} (1 + (hn)^2) |\psi_n| |\overline{y}_n|\nonumber \\ & \le 
	\|\psi\|_{l^2_w} \|y\|_{l^2_0} \nonumber\\ & \le
c_2\| \psi \|_{l^2_0} \|y\|_{l^2_0}\le c_2 M_{\mu,0,f}^{1/2}
\| \psi(0) \|_{l^2_0} \|y\|_{l^2_0}.
\end{align}
By using \eqref{distest3D}, instead of \eqref{ineqfin}, we get the inequality
\begin{eqnarray}
	\label{ineqfinD}
	 \frac{d}{dt} \|y(t)\|_{l^2_0}  \le  
	|g|  \|\phi (0)\|_ {l^2_0}^3  + 
	2\mu \tilde{C}_{\mu,0} \| \psi(0) \|_ {l^2_0} ^ 3  +
	\Omega^2 \tilde{C}_{1,\mu,0}
	\| \psi(0) \|_{l^2_0},\;\; 
\end{eqnarray} 
with $\tilde{C}_{\mu,0}=M_{\mu,0,f}^{3/2}$ and $\tilde{C}_{1,\mu,0}=c_2 M_{\mu,0,f}^{1/2}$. Integration of \eqref{ineqfinD} with respect to time, leads to the estimate \eqref{final1D} with the linear growth rate \eqref{final1aD}. The global, uniform-in-time bound for the distance is
\begin{eqnarray}
	\label{trans1d}
	\|y(t)\|_{l^2_0}\leq  \|\phi(0)\|_{l^2_0}+C^*_{1,\mu,0}\|\psi(0)\|_{l^2_0}:=\mathcal{B}[\phi(0),\psi(0)],\;\;\;C^*_{1,\mu,0}=M_{\mu,0,f}^{1/2}.
\end{eqnarray} 
We conclude with the proof of \eqref{finitD}-\eqref{finit2D}, by following the last lines of the proof of Theorem \ref{MR}, this time, by using the upper bounds \eqref{final1D} and \eqref{trans1d}. 
 \end{proof}
 %%%%%%%%%%%%%%%%%%%%5
\begin{figure}[tbp!]
	\centering 
	\begin{tabular}{cc}
		(a)&\hspace{0.7cm} (b)\\
		\includegraphics[scale=1]{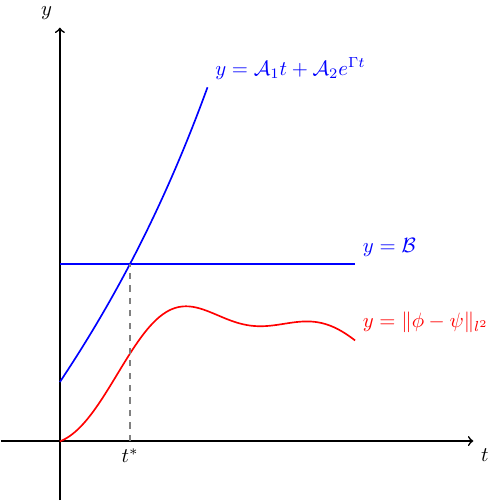}&
		\hspace{0.3cm}
		\includegraphics[scale=1]{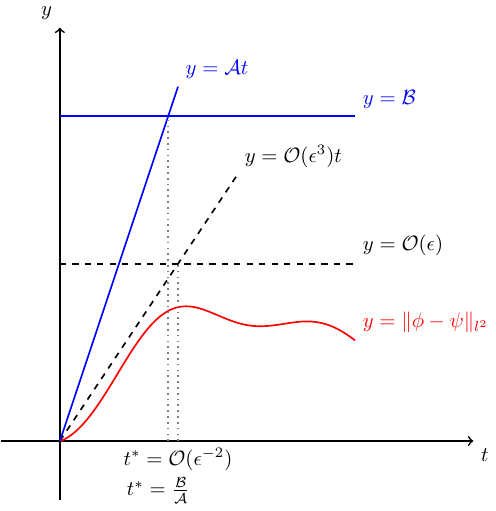}
	\end{tabular}
	\caption{
		\label{Fig1} (a) Geometrical interpretation of Theorem \ref{MR}. (b) Geometrical interpretation of Corollary \ref{MRF}. Details in the text (see Section \ref{SecIIIB}).}
\end{figure}
%%%%%%%%%%%%%%%%%%%%%%%%%%%%%%%%%%%%%%%%%%%%%%%%%%%%%%%%%%
\begin{figure}[tbp!]
	\centering 
	\begin{tabular}{cc}
		(a)&\hspace{0.7cm}(b)\\	
		\includegraphics[scale=0.6]{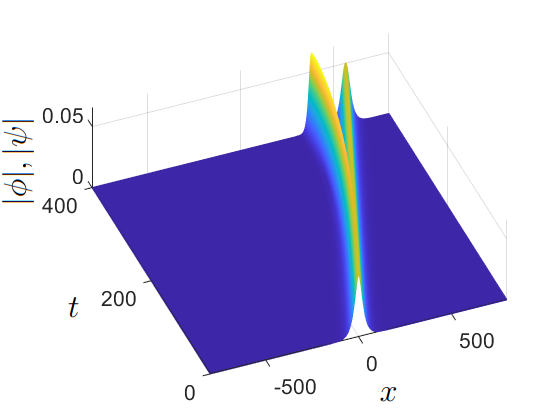}&
		\hspace{0.3cm}\includegraphics[scale=0.52]{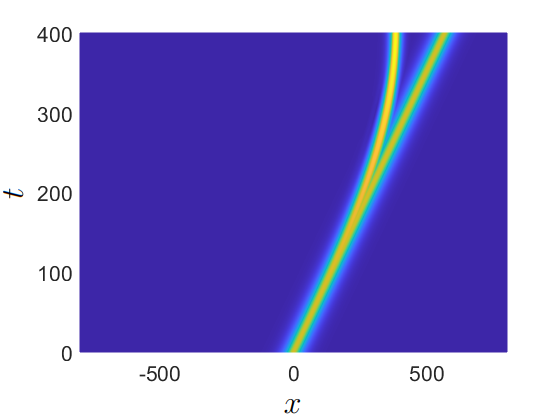}\\[5pt]
		(c)&\hspace{0.7cm}(d)\\
		\includegraphics[scale=0.5]{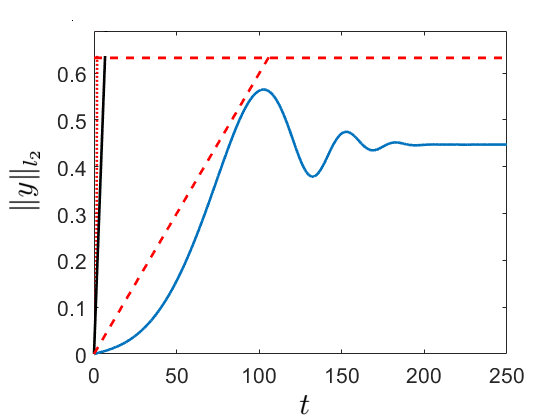}&
		\hspace{0.3cm}\includegraphics[scale=0.41]{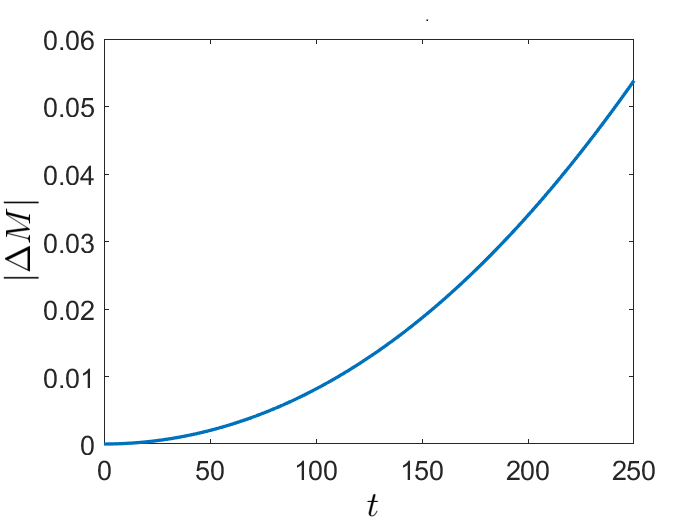}
	\end{tabular}
	\caption{
		\label{Fig2} (a) Spatiotemporal evolution   of $|\phi_n(0)|$ by the DGPE lattice, where $\phi_n(0)=\psi_n(0)$  is the initial condition defined by the analytical solution \eqref{ALBS},  against the same evolution by the AL lattice. The curved orbit of the soliton corresponds to the DGPE soliton while the straight orbit to the analytical AL soliton. (b) Contour plot of the evolution in panel (a). (c) The evolution of $\|y(t)\|_{l^2}$. (d) Evolution of the difference of the momenta $|\Delta M(t)|$ of the solutions.  Parameters: Amplitude of the initial condition $A=0.05$, $\Omega=0.002$. More details in text (see Section \ref{SecIV}-paragraph a.)}
\end{figure}
%%%%%%%%%%%%%%%%%%%%%
\begin{figure}[tbp!]
	\centering 
	\begin{tabular}{cc}
		(a)&\hspace{0.7cm}(b)\\	
		\includegraphics[scale=0.6]{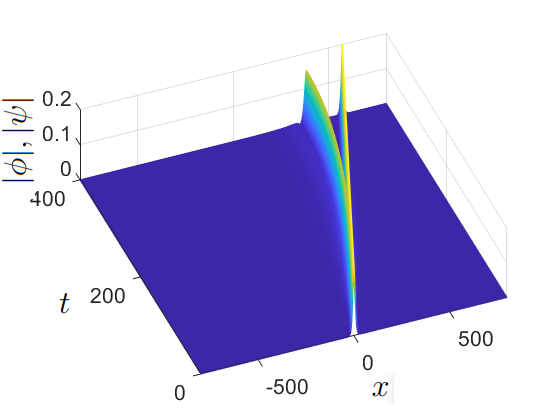}&
		\hspace{0.3cm}\includegraphics[scale=0.52]{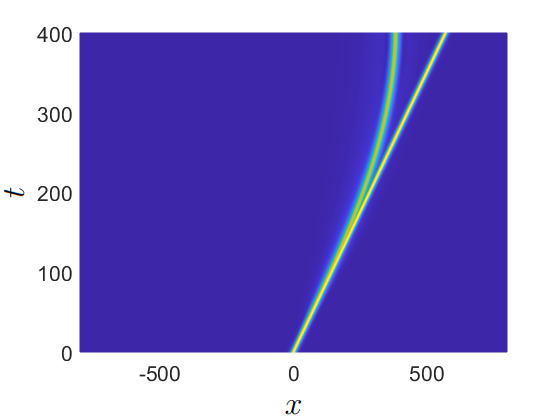}\\[5pt]
		(c)&\hspace{0.7cm}(d)\\
		\includegraphics[scale=0.5]{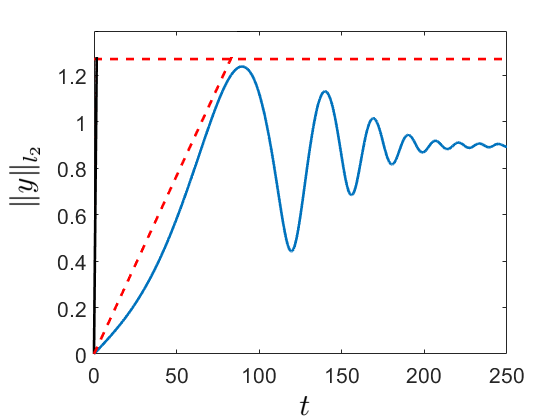}&
		\hspace{0.3cm}\includegraphics[scale=0.41]{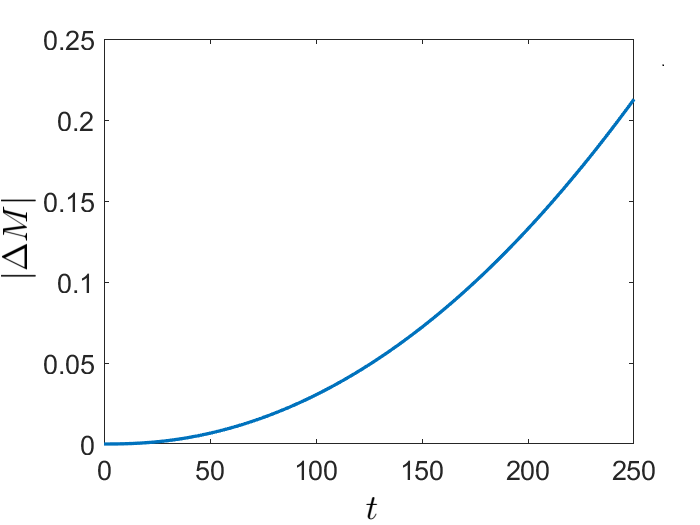}
	\end{tabular}
	\caption{
		\label{Fig3} The same numerical results as those presented in Figure \ref{Fig2}, but for $A=0.2$. More details in text (see Section \ref{SecIV}-paragraph b.)}
\end{figure}
%%%%%%%%%%%%%%%%%%%%%
\subsection{Understanding the dynamics of the distance}
\label{SecIIIB}
Theorem \ref{MR} for infinite lattices and Corollary \ref{MRF} for finite lattices provide insight into the potential deviation of the dynamics of the non-integrable DGPE system from those of the integrable AL system. First, this deviation is always finite and bounded by the system's initial conditions, as demonstrated by the global, uniform-in-time bounds \eqref{trans1} for the infinite lattice and \eqref{trans1d} for the finite lattice. Second, the deviation can grow only transiently, with at most an exponential rate for the infinite lattice and at most a linear rate for the finite lattice, as shown in estimates \eqref{final1}-\eqref{final1a} and \eqref{final1D}-\eqref{final1aD}, respectively. In other words, sub-exponential or sub-linear growth may occur transiently, as the relevant growth estimates increase only temporarily until they reach the global uniform bounds of the distance $\mathcal{B}$.

Thus, the geometric interpretation of Theorem \ref{MR} and Corollary \ref{MRF} is as follows: The growing upper bounds \eqref{final1}-\eqref{final1a} and \eqref{final1D}-\eqref{final1aD}, along with the global upper bounds $\mathcal{B}$, create a trapezoidal-like trapping region for the distance $\|y(t)\|_{\mathcal{X}}$ in the $(t,\|y\|_{\mathcal{X}})$-plane, where $\mathcal{X} = l^2$ for the infinite lattice and $\mathcal{X} = l^2_0$ for the finite lattice. This region is unbounded in the time direction since the solutions exist globally. The top horizontal side of the region is the global upper bound $\mathcal{B}$, while the bottom horizontal side is the $t$-axis. The left side of the trapezoidal region is formed by the graphs of the growing upper bounds \eqref{final1}-\eqref{final1a} for the infinite lattice or \eqref{final1D}-\eqref{final1aD} for the finite lattice. In the simplest case, where the systems start with identical initial conditions, i.e., $\|y(0)\|_{\mathcal{X}} = 0$, the graph for the infinite lattice is exponential, given by $y = \mathcal{A}_1t + \mathcal{A}_2 \exp\left(\Gamma t\right)$, while for the finite lattice, it is linear, $y = \mathcal{A}t$. These bounding curves are traced until they intersect the horizontal line $y = \mathcal{B}$ at times $t^*$, which are defined by \eqref{finit}-\eqref{finit2} for the infinite lattice and by \eqref{finitD}-\eqref{finit2D} for the finite lattice. They cannot extend above the horizontal line $y = \mathcal{B}$, since $\mathcal{B}$ represents the global upper bound of the distance.

This trapezoidal-like trapping region for the distance $\|y(t)\|_{\mathcal{X}}$ is illustrated in Figure \ref{Fig1}, with Panel (a) corresponding to Theorem \ref{MR} and Panel (b) to Corollary \ref{MRF}.

Additionally, let us remark the following: 
\begin{enumerate}
	\item
In the limit $\Omega \rightarrow 0$, where the DGPE lattice becomes the standard DNLS lattice, the left bounding curve becomes linear, since $\lim_{\Omega\rightarrow 0}\mathcal{A}_2 = 0$ (as follows from the expression of the coefficient $\mathcal{A}_2$ in \eqref{final1a}), eliminating the exponential term in the estimate. This case was examined in \cite{DJNa}, and also corresponds to the distance between the AL and generalized non-integrable AL lattices \cite{DJNb}, where the growing upper bounds were also linear. Another important example where the linear transient estimate is valid is the distance between the solutions of the integrable cubic NLS partial differential equation and the solutions of a broad class of non-integrable counterparts, supplemented with various boundary conditions in infinite or finite intervals \cite{JDE2024}.
%%%%%%%%%%%%%%%%%%%%%%%%%%%%%%%%%%%%%%%%%%
\item In the case of the finite lattice, due to the presence of the constant $c_2(N)$ in the equivalence of norms \eqref{normequiv}, in the discrete regime $h=\mathcal{O}(1)$, we have that $\lim_{N\rightarrow\infty}c_2(N)=\infty$, and consequently, for the slope $\mathcal{A}(N)$ of the linear estimate $y=\mathcal{A}t$ (through the dependence of $\mathcal{A}$ on $\tilde{C}_{1,\mu,0}=c_2 M_{\mu,0,f}^{1/2}$), we have $\lim_{N\rightarrow\infty}\mathcal{A}(N)=\infty$. In this case, the trapezoidal region defined by Corollary \ref{MRF} is formed by the vertical axis $y$ and the horizontal line $y=\mathcal{B}$. Evidently, such a behavior is a consequence of the presence of the harmonic trap of the model and of the fact that the AL lattice, lacks uniformly bounded solutions in time, in $\ell^2_w$.  On the other hand, in the case of a continuous limit $h=\mathcal{O}(N^{-1})$ (with the limitations stated in Remark \ref{DPDE} for  the consideration of such a limit), the linear estimate $y=\mathcal{A}t$ becomes quantitatively relevant since $\mathcal{A}$ is finite. Moreover, for both discrete and continuous regimes, in the case where $\Omega\rightarrow 0$, we have that $\mathcal{A}=\mathcal{O}( \|\phi(0)\|_ {l^2_0}^3,  \|\psi(0)\|_ {l^2_0}^3)$. 
\end{enumerate}
The further understanding of the dynamics of the distance $\|y(t)\|_{\mathcal{X}}$ provided in this paper, motivated by the DGPE lattice, and the relevant results of \cite{DJNa, DJNb, JDE2024} strongly support the argument that the description visualized in Figure \ref{Fig1} is fairly generic. The interpretation of Figure \ref{Fig1} holds for a wide class of non-integrable discrete and continuous NLS-type systems relative to their integrable counterparts. It establishes that the deviation of the dynamics of relevant non-integrable systems from the dynamics of their integrable counterparts, measured in the norms of their phase spaces $\mathcal{X}$, cannot be infinite but is bounded, with bounds depending on the size of the initial data.

In the case where the size of the initial data is small in the norm of $\mathcal{X}$, specifically of the order $\mathcal{O}(\varepsilon)$, for $0<\varepsilon<1$,  the corresponding trapping regions become narrower. This is illustrated by the region bounded by the dashed black lines in Figure \ref{Fig1} (b). The example demonstrates a scenario where $\mathcal{B} = \mathcal{O}(\varepsilon)$, and the transient linear upper bound is given by $y=\mathcal{O}(\varepsilon^3)t$, which is relevant for $\Omega\rightarrow 0$. This case encompasses significant examples of Hamiltonian NLS-type systems encountered thus far, as the DNLS, generalized AL systems, and non-integrable 1D NLS partial differential equations.

In such narrower regions, and for significant periods of time, we may expect dynamics very close to those of the integrable counterpart, as the numerical studies in \cite{DJNa, DJNb, JDE2024} illustrated.  This will also hold for the DGPE lattice relative to the AL, as we will discuss in the next section.

%%%%%%%%%%%%%%%%%%%%%5

\section{Numerical Results}
\label{SecIV}
In light of the theoretical results from Section \ref{SecIII}, we explore the potential proximity of wave dynamics between the DGPE and the AL lattice. We use the 1-soliton analytical solution of the AL lattice \eqref{ALBS} as initial conditions for the DGPE, focusing on small amplitudes and small values of $\Omega$. This approach allows us to numerically investigate the potential persistence of integrable bright soliton dynamics, as described by the analytical solution \eqref{ALBS} of the AL lattice, within the specified regime of initial conditions and harmonic trap strengths.

For the numerical studies, we fix the following parameters: $\alpha=\pi/4$, $\mu=1$ and $g=-1$ (focusing). We vary only the amplitude, denoted by the parameter $b$. We consider the systems in a finite lattice $[-800,800]$ with lattice spacing $h=1$. In the DGPE lattice, we set the small value $\Omega=0.002$.
\paragraph{The case of initial amplitude $A=0.05$.}
The initial condition defined by the soliton solution \eqref{ALBS} has the following norms: 
\begin{eqnarray}
	\label{nominis}
\varepsilon=	\|\psi(0)\|_{l^2}=0.316,\;\; \|\psi(0)\|_{l^2_w}=213.584.
\end{eqnarray}   
For the initial data and parameters as given above, with $\varepsilon=\mathcal{O}(10^{-1})$ and $\Omega=\mathcal{O}(10^{-3})$, the coefficients of the exponential bounding curve
 $y = \mathcal{A}_1t + \mathcal{A}_2 \exp\left(\Gamma t\right)$ for the trapping region are of the order $\mathcal{A}_1=\mathcal{O}(\epsilon^3)=\mathcal{O}(10^{-3})$ and $\mathcal{A}_2=\mathcal{O}(10^{-4})$. Specifically, we have $\Gamma = 4.2$, since $c_1=1$ for the given weight $w_n=(1+|hn|^2)^2$. Additionally, $\mathcal{B} \sim 2\varepsilon$, due to the fact that  $M_{\mu,0}=\mathcal{O}(1)$ with $M_{\mu,0} \lesssim 1$,  enabling for the approximation of the constants $C_{\mu,0} \sim 1$ and $C^*_{\mu,0} \sim 1$. The fact that the constant $M_{\mu,0}$ is of the above order for the given parameters and data, follows by applying Lemma \cite[Lemma 2.1 ]{DJNa} and the elementary inequality $\ln (1+x)\leq x$ for all $x>0$.
 Therefore, for the selected parameters, the trapezoidal trapping region is defined by the graphs of the functions 
\begin{eqnarray}
	\label{trap1}
y\sim 3\varepsilon^3t+	\frac{\Omega^2}{\Gamma} \| \psi(0) \|_{l^2_w}\exp\left(\Gamma t\right)=0.095t + 2\cdot10^{-4}\exp(4.2t)\;\;
\mbox{and the horizontal line}\;\;y=2\varepsilon=0.632.\;\;\;\;\;\;\;\;\;
\end{eqnarray}
Note that, due to the presence of the rapidly growing term $\exp(4.2t)$, the exponential bounding curve becomes steep after a short time. Thus, it is natural to consider the trapping region as formed by this exponential curve, which quickly becomes almost parallel to the vertical axis of the $(t, \|y\|)$-plane, and the horizontal line $y=2\varepsilon=0.632$. In this trapping region, with a height of order $\mathcal{O}(10^{-1})$, we may expect proximal dynamics between the systems for significant times.

Panels (a) and (b) of Figure \ref{Fig2} depict the dynamics of the AL soliton initial condition for $A=0.05$. The DGPE soliton follows a curved orbit $|\phi(t)|$, an effect of the nonzero $\Omega$, while the AL soliton follows a straight orbit $|\psi(t)|$. The solitons remain on similar orbits up to significant times, around $t \sim 200$. After this point, the orbits diverge due to the weak harmonic trap, which is strong enough to enforce the curved orbit of the DGPE soliton. Nevertheless, the DGPE soliton exhibits remarkable stability along its curved orbit. Panel (c) shows the evolution of the distance $\|y(t)\|_{l^2}$  represented by the continuous (blue) oscillating curve, further confirming the theoretical results, as $\|y(t)\|_{l^2}$, remains confined within the described trapping region \eqref{trap1} depicted a follows: The dotted (red) curve, close to the $\|y\|_{l^2}$-axis represents the steep exponential curve of \eqref{trap1}; in the plot it appears to be  depicted as almost linear.  The dashed (red) horizontal line represents the line $y=2\varepsilon$ of \eqref{trap1}.  The continuous (black) line which is less steeper than the dotted (red) curve, is the linear counterpart  $y=\mathcal{A}t=3\varepsilon^3t$ which corresponds to the case $\Omega\rightarrow 0$. We observe that the evolution of $\|y(t)\|_{l^2}$ demonstrates a transient, almost linear growth, at a much slower rate than the very small linear rate predicted by the bound $y=\mathcal{A}t=3\varepsilon^3t$, which explains the closely aligned dynamics observed at large times. Actually, the same numerical studies conducted in \cite{DJNa, DJNb, JDE2024} (results not shown here) confirmed that the rate of transient growth in the DGPE case  {\em remains of the order $\mathcal{O}(\varepsilon^{5})$}, which is close to the order identified in DNLS systems with $\Omega=0$; variations from this order may appear in short time intervals, within the time scales of the growth of $\|y(t)\|_{l^2}$, but never exceed the theoretical one of $\mathcal{O}(\varepsilon^3)$. To compare with the numerical rate of the transient linear divergence we also plotted  $y=\mathcal{A}_{\mathrm{num}}t$, with $\mathcal{A}_{\mathrm{num}}=3\varepsilon^5$, represented by the dashed (red) line and is close the  curve of the numerical distance $||y(t)||_{l^2}$. The global maximum of $\|y(t)\|_{l^2}$ is close to the predicted horizontal line $y=2\varepsilon=0.632$. The oscillations following this maximum and the eventual convergence to a constant value are due to the separation of the DGPE soliton from the AL soliton.  After their complete separation,  $\|y(t)\|_{l^2}$ becomes practically constant due to the conservation of the individual norms of the two solutions, see the discussion on Figure \ref{FigEX}.

Finally, recall that the momentum  $M(t)=2\mathrm{Im}\big\{\sum_{n\in\mathbb{Z}}\psi_n\overline{\psi}_{n+1}\big\}$ is conserved only in the AL lattice. By repeating the proof of \cite[Corollary III.1]{DJNb}, we can verify that, in the presence of the exponential bound described by \eqref{trap1},  the absolute value of the difference in momentum between the systems $|\Delta M|\sim\mathcal{O}(\varepsilon^3)$ for significant times. This provides further evidence of the proximity between the  dynamics of the two systems. The numerical results for the evolution $|\Delta M (t)|$, shown in panel (d) of Figure \ref{Fig2}, are consistent with this theoretical prediction. 

Note that in Figure \ref{Fig2},  we considered two different time scales for the plots. In the top panels (a) and (b), we considered times up to $t=400$, in order to visualize better the separation of the solitons, as well as the robust evolution of the DGPE soliton after the separation. In the bottom panels (c) and (d), we plotted the evolution of the distance $\|y(t)\|_{l^2}$ up to  $t=250$ because this is the actual time at which the separation of the tails of the solution occurs, and after that time, the comparison of the dynamics between the two systems is not meaningful. This way of depicting the numerical results will be followed also in the next studies. 
%%%%%%%%%%%%%%%%%%%%%%%%%%%%%%%%%%%%%%%%%%%%%%%%%%%%%%
\paragraph{The case of initial amplitude $A=0.2$.} The initial condition defined by the soliton solution \eqref{ALBS} has the following norms: 
\begin{eqnarray}
	\label{nominis2}
	\varepsilon=	\|\psi(0)\|_{l^2}=0.634,\;\; \|\psi(0)\|_{l^2_w}=27.4.
\end{eqnarray}   
For these parameters, $\Gamma=4.8$ and  the trapping region is described by the graphs of the almost boxed trapping region, defined by the functions 
\begin{eqnarray}
	\label{trap2}
	y\sim 3\varepsilon^3t+	\frac{\Omega^2}{\Gamma} \| \psi(0) \|_{l^2_w}\exp\left(\Gamma t\right) = 0.77t + 2.28\cdot10^{-5}\exp(4.8t)\;\mbox{and the horizontal line}\;\;y=2\varepsilon=1.268.\;\;\;\;\;\;\;
\end{eqnarray}
The transient growth estimate predicts an increase in the growth rate of $\|y(t)\|_{l^2}$ and its global upper bound $\mathcal{B}$. Panel (c) of Figure \ref{Fig3} shows that the distance $\|y(t)\|{l^2}$ remains confined within the trapping region \eqref{trap2}, and the global upper bound $\mathcal{B}$ is shown to be quite sharp. While the transient growth rate of $\|y(t)\|_{l^2}$ increases, qualitatively aligning with the theoretical results, the numerical growth rate is still much smaller than the predicted theoretical value.  As a result, this increase of the  amplitude of the initial conditions does not significantly alter the proximal dynamics of the DGPE in relation to the AL lattice. The numerical results shown in Figure \ref{Fig3} closely resemble those reported for the previous case with $A=0.05$, as illustrated in the top panels (a) and (b) of Figure \ref{Fig3}, and panel (c) for $|\Delta M (t)|$. Amplitude variations within the regime of proximal dynamics are expected to be even smaller, due to the embedding \eqref{ineq_l}, which implies that $\|y(t)\|_{l^{\infty}}\leq \|y(t)\|_{l^2}$.
%%%%%%%%%%%%%%%%%%%%%%%%%%%%55555
\begin{figure}[tbp!]
	\centering 
	\begin{tabular}{cc}
		(a)&\hspace{1cm}(b)\\	
		\includegraphics[scale=0.5]{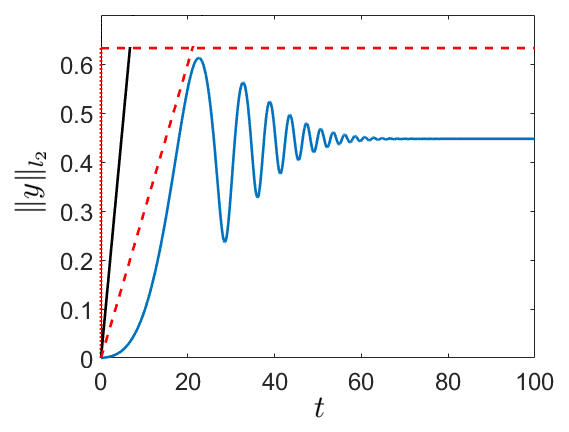}&\includegraphics[scale=0.5]{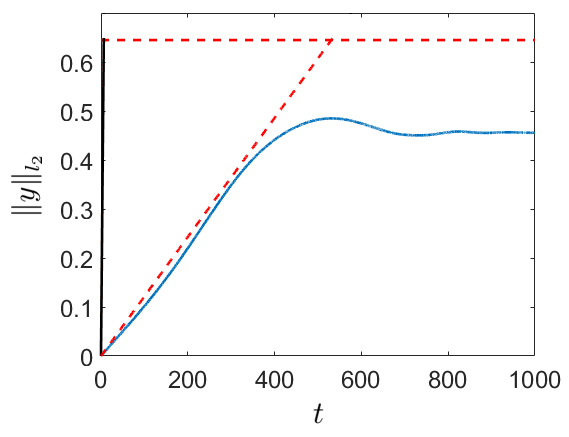}
	\end{tabular}
	\caption{
		\label{Fig5} The same numerical results as those presented in Figure \ref{Fig2}-panel (c), but for two different values of $h$. (a)  $h=0.1$.  (b) $h=10$. More details in text (see Section \ref{SecIV}-paragraph c).}
\end{figure}
%%%%%%%%%%%%%%%%%%%%%%%%%%%%%%%%%
\paragraph {Variations of the parameter $h$.} We conclude the numerical studies, by considering the effect of the variation of $h$ in the proximity dynamics. We consider two representative cases for variable $h$, namely $h=0.1$, which corresponds to the case of a continuous limit, and $h=10$, which corresponds to the case of a strong discrete regime, that is, the anti-continuous limit. For both cases, the amplitude is fixed to $A=0.05$, and the strength of the trap is small, $\Omega=0.002$. The results of these simulations, for the behavior of the distance $\|y(t)\|_{l^2}$ and its theoretical estimates, are depicted in Fig.~\ref{Fig5}.

Regarding the proximal dynamics, we observe that when $h=0.1$, the separation of the two solutions occurs faster, and the dynamics remain proximal for shorter times, up to $t\sim 100$. The faster separation of the orbits is a direct consequence of the dependence of the velocity $v$ of the soliton solution \eqref{ALBS} on $h$, according to $v\propto h^{-2}$. It is interesting to observe that the linear estimate $y=\mathcal{A}t=3\varepsilon^3t$ [represented again by the black (continuous) line] becomes quantitatively relevant, in agreement with the discussion in  \ref{SecIIIB}-comment 2: indeed, for  lattice paramaters as in this example with $N=\mathcal{O}(10^3)$ and $h=\mathcal{O}(10^{-1})$, the constant of equivalence of norms given in \eqref{normequiv} is  $c_2(N)\approx \tilde{C}_{1,\mu,0}=\mathcal{O}(10^4)$, while $\Omega^2=\mathcal{O}(10^{-6})$ and $\varepsilon= \| \psi(0)\| _{l^2}=\mathcal{O}(10^{-1})$.  Then, the theoretical rate of the linear, transient divergence given in \eqref{final1aD} is $\mathcal{A}=\mathcal{O}(10^{-3})=\mathcal{O}(\varepsilon^3)$. The numerical rate is found to be  $\mathcal{A}_{\mathrm{num}}=\mathcal{O}(\varepsilon^4)$.  The dashed (red) line to the left of the continuous (black) line represents $y=\mathcal{A}_{\mathrm{num}}t=3\varepsilon^4t$.  

Accordingly, when $h=10$, the separation of the solitons occurs much later, at $t\sim 400$. The linear estimate $y=\mathcal{A}t$ has a significantly large rate $\mathcal{A}$, yet, in agreement with the discussion in  \ref{SecIIIB}-comment 2, while the numerical rate $\mathcal{A}_{\mathrm{num}}=3\varepsilon^{6.9}$. This time, the line $y=\mathcal{A}_{\mathrm{num}}t=3\varepsilon^{6.9}t$ is almost tangent to the distance $\Vert y(t)\Vert _{l^2}$ for the time scales of its linear growth rate.

%%%
\section{Conclusions}
\label{SecV}
Continuing our investigation into the proximity of wave dynamics in non-integrable NLS-type models to those of their integrable counterparts \cite{DJNa,DJNb,JDE2024}, specifically the structural stability of the relevant integrable systems \cite{book2}, we explored the case of the focusing discrete Gross--Pitaevskii equation against the Ablowitz--Ladik lattice. This motivated a deeper study and understanding of the dynamics of the distance between solutions of these systems. The description appears to be quite general and applies to a wide class of non-integrable models.

When the norms of the solutions in the relevant phase spaces are bounded by quantities that depend on the norms of the initial conditions — a property shared by many systems in both discrete and continuous settings — the distance, measured in the associated metric, remains bounded. It is confined to regions defined by the global, uniform-in-time bound of the distance and by curves representing upper, time-dependent estimates for the transient growth of the distance, which cannot exceed the global bound. These trapezoidal-like regions are unbounded only in the time direction, as the solutions of the systems exist globally in time. In the regime of small data and parameters that may perturb the integrable model, these regions are predicted to be narrow. Within these regions, the dynamics of the integrable model are expected to persist in the non-integrable model over long periods, potentially for significant times. This persistence is due to the fact that the upper estimates for the transient growth of the distance suggest at most a linear increase. For initial data with norms of the order $\mathcal{O}(\varepsilon)$, where $0<\varepsilon<1$, the upper estimates for the transient growth of the distance, which in most cases are linear with respect to time, predict a small growth rate that is at most of the order $\mathcal{O}(\varepsilon^3)$.

The study of the discrete Gross--Pitaevskii equation introduced novel implications for the approach to proving theoretical estimates. Due to the presence of the harmonic trap, the analysis required the study of the AL lattice in weighted spaces, as the weighted norms for its solutions play a role in the evolution equation for the distance. It should be remarked that the problem of studying the Ablowitz--Ladik lattice in weighted spaces is important due to its strong connection with the applicability of the inverse scattering transform and the determination of various decaying properties of the solutions initiated by general initial data \cite{Yamane1}, \cite{China1}, \cite{Yamane2}, \cite{China2}.  With this approach, the transient growth estimates suggest sub-exponential growth for the distance in the case of an infinite lattice or sub-linear growth in the case of a finite lattice with Dirichlet boundary conditions. For small data and parameters, the previously mentioned trapping regions for the distance are predicted to be narrow, further suggesting the proximity of the dynamics between the systems.

Numerical studies with small initial data and strengths of the harmonic trap, aimed at assessing the potential persistence of the analytical Ablowitz--Ladik one soliton solution in the dynamics of the discrete Gross--Pitaevskii equations, illustrate agreement with the theoretical predictions. The numerical growth rate of the distance was found to be considerably smaller than the already small theoretical prediction, being approximately of the order $\mathcal{O}(\varepsilon^{5})$. Thus, the soliton excited by the discrete Gross--Pitaevskii equation, using the analytical bright soliton solution of the Ablowitz--Ladik lattice as its initial condition, exhibits remarkable proximity to the dynamics of the AL analytical solution over long times. An interesting finding is that even beyond these long times, when the influence of the weak harmonic trap becomes more significant, the soliton still demonstrates remarkable robustness along its curved orbit.

	The above results motivate further studies. One interesting direction is
investigating the potential proximity of dynamics for general localized initial
data, exploiting the detailed long-time asymptotics of the Ablowitz--Ladik
lattice~\cite{Yamane1} and the comparison of the systems under a potential common symplectic structure \cite{R2}. Another challenging area is the defocusing counterpart
with nonzero boundary conditions at infinity, which relates to the emergence of
dark solitons in the system. Similarly, the focusing case with nonzero boundary
conditions is also challenging due to its connection to rogue wave solutions \cite{R1,R4}.
Ongoing investigations explore these problems in both discrete and continuous
settings, with results to be reported in future work.
%\newpage
\begin{acknowledgments}
The paper is dedicated to the memory of Noel F.\,Smyth.	
\end{acknowledgments}
\vspace{0.5cm}
\noindent
\textbf{Declarations Statement}\\
The authors declare that they have no known competing financial interests or personal relationships that could have appeared to influence the work reported in this paper.\\
\\
\textbf{Authors Contributions Statement}\\
All authors contributed equally to the study conception, design and writing of the manuscript. Material preparation, data collection and analysis were performed equally by all authors.  All authors read and approved the final manuscript.
%%%%%%%%%%%
%%%%%%%%%
%%%%%%%%%%%%%%%%%%%%%%%%%%%
%%%%%%%%%%%%%%%%%%%%%%%%%5
%%%%%%%
%%%%%%%%%%%
%\newpage
%%%%

\end{document}